\pdfoutput=1
\RequirePackage{ifpdf}
\ifpdf 
\documentclass[pdftex]{sigma}
\else
\documentclass{sigma}
\fi

\numberwithin{equation}{section}

\newtheorem{thm}{Theorem}[section]
\newtheorem{cor}[thm]{Corollary}
\newtheorem{lem}[thm]{Lemma}

\theoremstyle{definition}

\newtheorem{exmp}[thm]{Example}

\newtheorem*{rem}{Remark}

\let\frak\mathfrak

\def\>{\relax\ifmmode\mskip.666667\thinmuskip\relax\else\kern.111111em\fi}
\def\<{\relax\ifmmode\mskip-.333333\thinmuskip\relax\else\kern-.0555556em\fi}
\def\vsk#1>{\vskip#1\baselineskip}
\def\vv#1>{\vadjust{\vsk#1>}\ignorespaces}
\def\vvn#1>{\vadjust{\nobreak\vsk#1>\nobreak}\ignorespaces}
 \let\alb\allowbreak

\let\mc\mathcal
\let\nc\newcommand

\let\al\alpha

\let\dl\delta

\let\ka\kappa

\let\phi\varphi

\let\der\partial

\let\ox\otimes

\let\geq\geqslant

\let\leq\leqslant

\let\on\operatorname

\def\F{{\mathbb F}} 

\def\+#1{^{\{#1\}}}
\def\lsym#1{#1\alb\dots\relax#1\alb}
\def\lc{\lsym,}

\def\End{\on{End}}

\def\be{\begin{equation*}}
\def\ee{\end{equation*}}

\def\slt{{\frak{sl}_2}}
\def\K{{\mathbb K}}
\def\Sing{\operatorname{Sing}}
\nc\Tt{{\bf T}}

\begin{document}

\allowdisplaybreaks

\newcommand{\arXivNumber}{2503.06326}

\renewcommand{\PaperNumber}{001}

\FirstPageHeading

\ShortArticleName{Finding All Solutions of qKZ Equations in Characteristic $p$}

\ArticleName{Finding All Solutions of qKZ Equations\\ in Characteristic $\boldsymbol{p}$}

\Author{Evgeny MUKHIN~$^{\rm a}$ and Alexander VARCHENKO~$^{\rm b}$}

\AuthorNameForHeading{E.~Mukhin and A.~Varchenko}

\Address{$^{\rm a)}$~Department of Mathematical Sciences, Indiana University Indianapolis, \\
\hphantom{$^{\rm a)}$}~402 North Blackford St, Indianapolis, IN 46202-3216, USA}
\EmailD{\mail{emukhin@iu.edu}}

\Address{$^{\rm b)}$~Department of Mathematics, University of North Carolina at Chapel Hill, \\
\hphantom{$^{\rm b)}$}~Chapel Hill, NC 27599-3250, USA}
\EmailD{\mail{anv@email.unc.edu}}
\URLaddressD{\url{http://varchenko.web.unc.edu}}

\ArticleDates{Received September 23, 2025, in final form December 16, 2025; Published online January 02, 2026}

\Abstract{In [\textit{J.~Lond. Math. Soc.}\ \textbf{109} (2024), e12884, 22~pages], the difference qKZ equations were considered modulo a prime number~$p$ and a family of polynomial solutions of the qKZ equations modulo $p$ was constructed by an elementary procedure as suitable $p$-approximations of the hypergeometric integrals. In~this paper, we study in detail the first family of nontrivial examples of the qKZ equations in characteristic~$p$. We describe all solutions of these qKZ equations in characteristic $p$ by demonstrating that they all stem from the $p$-hypergeometric solutions. We also prove a~Lagrangian property (called the orthogonality property) of the subbundle of the qKZ bundle spanned by the $p$-hypergeometric sections. This paper extends the results of [arXiv:2405.05159] on the differential KZ equations to the difference qKZ equations.}


\Classification{11D79; 12H25; 32G34; 33C05; 33E30}

\section{Introduction}

The Knizhnik--Zamolodchikov (KZ) differential equations are a system of
linear differential equations, satisfied by conformal
blocks on the sphere in the WZW model of conformal field theory, see \cite{KZ}.
The quantum Knizhnik--Zamolodchikov (qKZ) equations are a difference version of the KZ equations which naturally appear in the representation theory of Yangians (rational case) and quantum affine algebras (trigonometric case), see \cite{EFK,FR}.
The qKZ equations may be regarded as a deformation of the KZ differential equations.

As a rule one considers the KZ and qKZ equations over the field of complex numbers. Then these differential and difference equations are solved in multidimensional hypergeometric integrals.

In \cite{SV},
the differential KZ
equations were considered modulo a prime integer $p$. It turned out that modulo $p$ the KZ equations have a family of polynomial solutions. The construction of these solutions was analogous to the construction of the multidimensional hypergeometric
solutions, and these polynomial solutions were called the $p$-hypergeometric
solutions.

In \cite{MV}, the rational $\slt$ qKZ equations with values in the $n$-th tensor power of the vector representation
$L$ and an integer step $\kappa$ were considered modulo $p$. A family of polynomial solutions modulo $p$
of these equations was constructed and called the $p$-hypergeometric solutions.

In this paper, we address the problem of whether all solutions of the qKZ equations in characteristic $p$
 are generated by the $p$-hypergeometric solutions. We consider the first family of nontrivial examples of the qKZ
 equations and demonstrate that, indeed, in this case, all solutions of the qKZ equations
 stem from the $p$-hypergeometric solutions.

Let $\K$ be a field of characteristic $p$. The qKZ equations for a function $f(z_1,\dots,z_n)$ with values
in the $\K$-vector space $L^{\otimes n}$ and step $\ka\in\K^\times$ have the form
\[
f(z_1,\dots, z_a-\ka,\dots, z_n) =
K_a(z;\ka) f(z),\qquad a=1,\dots,n,
\]
where the linear operators $K_a(z;\ka)$ are given in terms of the rational $\slt$ $R$-matrix, see \eqref{K}.
The operators $K_a(z;\ka)$ commute with the diagonal action of $\slt$,
 and, therefore, it is sufficient to solve the qKZ equations only with values in the space of singular vectors of a given weight.
In~this paper, we study the qKZ equations with values in
$V:=\Sing L^{\ox n}[n-2] \subset L^{\ox n}$, the subspace of singular vectors of weight $n-2$. We have $\dim V=n-1$.

There are two cases: $\ka\in \K\setminus \F_p$ and $\ka\in\F_p^\times$.

\begin{thm}\label{thm 1.1}
Let $p$ be a prime number that does not divide $n$.
For $\ka\in\K\setminus \F_p$, there does not exist a nonzero rational $V$-valued function $f(z_1,\dots,z_n)$ which is a solution of the qKZ equations with parameter $\ka$.
\end{thm}

See Corollary \ref{cor ne F}.

Assume that \smash{$\ka\in\F_p^\times$}.
Let $0<k<p$ be the positive integer such that~${\ka k \equiv -1 \pmod{p}}$. Let $[x]$ denote the integer part of a real number and
$d(\ka) :=\big[\frac{kn}{p}\big]$.
If $p$ does not divide $n$, then~${d(\ka)+d(-\ka)=n-1 = \dim V}$.

In \cite{MV}, we constructed $d(\ka)$
$V$-valued $p$-hypergeometric solutions of the qKZ equations denoted~${Q^{\ell p-1}(z;\ka)}$, $\ell = 1, \dots,d(\ka)$. In this paper, we show that these solutions are linearly independent
over the field $\K(z_1,\dots,z_n)$, see Theorem \ref{thm ind}.

\begin{thm}\label{thm main in}
Let $p>n$, $\ka\in\F_p^\times$, and $0<d(\ka)<n-1$.
Let $f(z)$ be a $V$-valued rational function in $z$ which is a solution
 the qKZ equations with step $\ka$.
 Then $f(z)$ is a linear combination of the $p$-hypergeometric solutions
$Q^{\ell p-1}(z;\ka)$, $\ell = 1, \dots, d(\ka)$, with coefficients which are scalar rational functions
in $z_i^p-z_i$, $i=1,\dots, n$.
\end{thm}

See Theorem \ref{thm main}. Notice that $h(x) = x^p-x \in\K[x]$ is a 1-periodic polynomial,
$h(x+1)=h(x)$. In particular, $h(x+\kappa)=h(x)$.

If $d(\ka)=n-1$ or $0$, all solutions of the qKZ equations with values in $V$ and step $\ka$ are described
in Section \ref{sec 7.4}.

We prove the orthogonality relations
for $p$-hypergeometric solutions of the qKZ
equations with steps $\ka$ and $-\ka$.

\begin{thm}\label{thm orth int}
Let $p>n$
and $0< d(\ka) <n-1$.
 Then for any $\ell \in \{1, \dots,d(\ka)\}$ and
$m \in \{1, \dots, d(-\ka)\}$,
we have
\[
\label{ort i}
S\big(Q^{m p-1}(-z;-\ka), Q^{\ell p-1}(z;\ka)\big)
= 0,
\]
where $S$ is the Shapovalov form.
\end{thm}

See Theorem \ref{thm orth}.

Define the $p$-curvature operators of the qKZ equations with values in $V$ and step $\ka$ by the formula
\begin{align*}
&C_a(z_1,\dots,z_n;\ka) :=
 K_a(z_1,\dots, z_a-(p-1)\ka,\dots, z_n;\ka)
 \\
 &\hphantom{C_a(z_1,\dots,z_n;\ka) := }{}\times
 K_a(z_1,\dots, z_a-(p-2)\ka,\dots, z_n;\ka)\cdots K_a(z_1,\dots, z_a-\ka,\dots, z_n;\ka)
\\
&\hphantom{C_a(z_1,\dots,z_n;\ka) := }{}\times
K_a(z_1,\dots, z_a,\dots, z_n;\ka),
\end{align*}
for $a=1,\dots,n$, and
 the reduced $p$-curvature operators by the formula
 \[
 \hat C_a(z_1,\dots,z_n;\ka) := C_a(z_1,\dots,z_n;\ka) -1.
 \]
If $f(z_1,\dots,z_n)$ is a solution of the qKZ equations, then $\hat C_a f=0$ for all $a$.

\begin{thm}
\label{thm cu in}

Let $p>n$.
Then the reduced $p$-curvature operators have the following properties:
\begin{enumerate}\itemsep=0pt

\item[$\on{(i)}$]

If $\ka\in\K\setminus \F_p$, then all reduced $p$-curvature operators $\hat C_a(z;\ka)$,
$a=1,\dots,n$, are nondegenerate for generic $z$.

\item[$\on{(ii)}$] If $\ka\in\F_p^\times$ and $d(\ka)=n-1$ or $0$, then all reduced $p$-curvature operators $\hat C_a(z;\ka)$,
$a=1,\dots,n$, are equal to zero.

\item[$\on{(iii)}$] If $\ka\in\F_p^\times$ and $0<d(\ka)<n-1$, then all reduced $p$-curvature operators $\hat C_a(z;\ka)$,
$a=1,\dots,n$, are nonzero. For every $a$,
the span of the $p$-hypergeo\-metric solutions
\smash{$Q^{\ell p-1}(z;\ka)$}, $\ell = 1, \dots, d(\ka)$, lies in the kernel of $\hat C_a(z,\ka)$ and contains the image of $\hat C_a(z;\ka)$. Also,
for all $a$, $b$,
\begin{gather*}
\hat C_a(z;\ka)\hat C_b(z;\ka)
= 0,
\qquad
\hat C_a(z,-\ka) +\hat C_a(-z;\ka)^*
= 0,
\end{gather*}
where for an operator $T\colon V\to V$ we denote by $T^*$ the operator dual to $T$ under the Shapovalov form.

\end{enumerate}

\end{thm}

See Theorem \ref{thm p-cu} and Lemmas \ref{lem c=c}, \ref{lem ne F}.

In a suitable limit the difference qKZ equations on $L^{\ox n}$ degenerate to the differential KZ equations on
$L^{\ox n}$,
\[
\ka \frac{\der f}{\der z_a} = \sum_{j\ne a} \frac{P^{(a,j)}-1}{z_a-z_j} f, \qquad a=1,\dots,n,
\]
where $P^{(a,j)}$ is the permutation operator of the $a$-th and $j$-th tensor factors of $L^{\ox n}$.

In \cite{VV1}, the differential KZ equations over a field $\K$ of characteristic $p$ with values
in $V \subset L^{\ox n}$ were studied in detail.
Our paper extends the results of \cite{VV1} from the differential KZ equations to the difference qKZ equations.
The proofs of Theorems \ref{thm 1.1}, \ref{thm main in}, and \ref{thm cu in} are based on the corresponding
results in \cite{VV1} for the differential KZ equations.

On the differential and difference equations in characteristic $p$ and associated $p$-curvature see also
\cite{EV, VV2}.

\section{Difference qKZ equations}
\label{sec 2}

\subsection{Notations}

In this paper, $p$ is a prime and $\K$ a field of characteristic $p$.

Consider the Lie algebra $\slt$ over $\K$ with basis $e$, $f$, $h$ and relations
$[e,f]=h$, $[h,e]=2e$,
$ [h,f]=-2f$.
Let $L$ be the two-dimensional $\slt$-module with basis $v_1$, $v_2$ and the
action
$ev_1=0$, $ev_2=v_1$, $fv_1=v_2$, $fv_2=0$, $hv_1=v_1$, $hv_2=-v_2$.

For a positive integer $n>1$, consider the $\slt$-module $L^{\ox n}$.

Let $\mc I_l$ be the set of all $l$-element subsets of $\{1,\dots,n\}$. For a subset
$I\subset \{1,\dots,n\}$, denote
\[
v_I=v_{i_1}\ox\dots\ox v_{i_n} \in L^{\ox n},
\]
where $i_j=2$ if $i_j\in I$ and $i_j=1$ if $i_j\notin I$.
Denote by $L^{\ox n}[n-2l]$ the span of the vectors $\{v_I \mid I\in\mc I_l\}$.
 We have a direct sum decomposition,
\[
L^{\ox n}=\bigoplus_{l=0}^n L^{\ox n}[n-2l].
\]
Let $\Sing L^{\ox n}[n-2l]\subset L^{\ox n}[n-2l]$ be the subspace of
singular vectors (the vectors annihilated by $e$).

\subsection{qKZ equations}
\label{sec qkz}

Define the rational $ R$-matrix acting on $ L^{\ox 2}$,
$
R(u) = \frac{u-P}{u-1}$,
where $P$ is the permutation of tensor factors of $L^{\ox 2}$.
The $R$-matrix satisfies the Yang--Baxter and unitarity equations,
\begin{gather}
R^{(12)}(u-v)R^{(13)}(u)R^{(23)}(v) =
R^{(23)}(v)R^{(13)}(u)R^{(12)}(u-v),\nonumber
\\
R^{(12)}(u)R^{(21)}(-u) = 1 .\label{unit}
\end{gather}
The first equation is an equation in $\End\big(L^{\ox 3}\big)$.
The superscript indicates the factors of $L^{\ox 3}$
on which the corresponding operators act.

Let $ z=(z_1\lc z_n)$.
Define the qKZ operators $ K_1,\dots, K_n\>$ acting on $L^{\ox n}$
\begin{align}
&K_a(z;\ka) =
R^{(a,a-1)}(z_a-z_{a-1}-\ka) \cdots R^{(a,1)}(z_a-z_1-\ka)\nonumber
\\
&\hphantom{K_a(z;\ka) =}{}
\times
R^{(a,n)}(z_a-z_n) \cdots R^{(a,a+1)}(z_a\<-z_{a+1}),\label{K}
\end{align}
where $\ka\in\K^\times$ is a parameter.

The qKZ operators preserve the weight decomposition of $L^{\ox n}$,
commute with the $\slt$-action,
and form a discrete flat connection with step $\ka$ on the trivial bundle $L^{\ox n}\times \K^n \to \K^n$,
\[
K_a(z_1,\dots, z_b-\ka,\dots, z_n;\ka) K_b(z;\ka)
=K_b(z_1,\dots, z_a-\ka,\dots, z_n;\ka) K_a(z;\ka)
\]
for $a,b=1,\dots,n$, see \cite{FR}.

The system of difference equations with step $\ka$,
\begin{equation}
\label{Ki}
s(z_1,\dots, z_a-\ka,\dots, z_n) =
K_a(z ;\ka) s(z),\qquad a=1,\dots,n,
\end{equation}
for an $ L^{\ox n}$-valued
function $s(z)$, is called the qKZ equations with step $\ka$.

Since the qKZ operators commute with the action of $\slt$ on $L^{\otimes n}$, the qKZ operators preserve the subbundle
$ \Sing L^{\ox n}[n-2l]\times \K^n\to \K^n$ for every integer $l$.

Define the translation operators $\bf T_a$ by
\[
({\bf T}_af)(z_1,\dots,z_a,\dots, z_n) = f(z_1,\dots, z_a-\ka, \dots, z_n).
\]
The difference operators
$
\nabla_a = \Tt_a^{-1} K_a(z;\ka)
$
 are called the connection
operators of the qKZ difference connection. We have $[\nabla_a,\nabla_b]=0$.

\subsection[p-curvature of the qKZ connection]{$\boldsymbol{ p}$-curvature of the qKZ connection}

 Define the $p$-curvature operators of the qKZ connection by
\begin{align*}
&C_a(z;\ka)=
 K_a(z_1,\dots, z_a-(p-1)\ka,\dots, z_n;\ka)K_a(z_1,\dots, z_a-(p-2)\ka,\dots, z_n;\ka)\cdots
\\
&\hphantom{C_a(z;\ka)=}{}
\times K_a(z_1,\dots, z_a-\ka,\dots, z_n;\ka)
K_a(z_1,\dots, z_a,\dots, z_n;\ka)
\end{align*}
for $a=1,\dots,n$. In other words,
$C_a = (\nabla_a)^p $.

For every $a$, the operator $C_a(z;\ka)$ acts on fibers
 of the bundle
 $L^{\ox n}\times \K^n\to \K^n$ and defines an endomorphism of the qKZ connection,
\[
K_b(z_1,\dots, z_n;\ka) C_a(z_1,\dots, z_n;\ka)
=
C_a(z_1,\dots, z_b-\ka,\dots, z_n;\ka) K_b(z_1,\dots, z_n;\ka).
\]

The operators $C_a(z;\ka)$ commute, i.e., $[C_a(z;\ka),C_b(z;\ka)]=0$.

 If $s(z;\ka)$ is a flat section of the qKZ discrete connection,
\[
s(z_1,\dots, z_a - \ka,\dots,z_n;\ka) = K_a(z;\ka)s(z;\ka), \qquad a=1,\dots,n,\]
then $s(z;\ka)$ is an eigenvector of the $p$-curvature operators with eigenvalue 1,
\begin{gather}
\label{fCf}
s(z;\ka) =C_a(z;\ka)s(z;\ka), \qquad a=1,\dots,n.
\end{gather}

An operator $C_a(z;\ka)$ is a rational function in $z$ with the denominator
\begin{align*}
D_a(z;\ka)
={}&
 \prod_{j\ne a} \prod_{m=0}^{p-1}(z_a-z_j-m\ka -1)
\\
={}&
 \prod_{j\ne a} \big(z_a^p-\ka^{p-1}z_a +( -z_j)^p+\ka^{p-1}z_j +(-1)^p+\ka^{p-1}\big).
\end{align*}
It is convenient to introduce the reduced $p$-curvature operators by the formula
\begin{gather}
\label{red p-cu}
\hat C_a(z;\ka) = C_a(z;\ka) -1,
\end{gather}
and the normalized $p$-curvature operators by the formula
\[
\tilde C_a(z;\ka) = D_a(z;\ka) (C_a(z;\ka) -1).
\]
The normalized $p$-curvature operators are polynomials in $z$ of degree $\leq (n-1) p$.

\subsection{Differential KZ equations}

For $\ka\in\K^\times$, the differential KZ operators
\[
\nabla_a^{\on{KZ}}
= \ka \frac{\der }{\der z_a} - \sum_{j\ne a} \frac{P^{(a,j)}-1}{z_a-z_j}, \qquad a=1,\dots,n,
\]
define a flat KZ connection on $L^{\ox n}\times K^n\to\K^n$, $\big[\nabla_a^{\on{KZ}}, \nabla_b^{\on{KZ}}\big]=0$.
The operators $\nabla_a^{\on{KZ}}$ commute with the
 $\slt$-action on $L^{\ox n}$. The system of equations
\begin{gather}
\label{KZ}
\ka \frac{\der f}{\der z_a} = \sum_{j\ne a} \frac{P^{(a,j)}-1}{z_a-z_j} f, \qquad a=1,\dots,n,
\end{gather}
 is called the differential KZ equations with parameter $\ka$.
The KZ operators preserve every subbundle $\Sing L^{\ox n}[n-2l]\times \K^n\to\K^n$.
Denote
\[
H_a(z) =
\sum_{j\ne a} \frac{P^{(a,j)}-1}{z_a-z_j},
\]
the Gaudin Hamiltonians.

The $p$-curvature operators of the KZ connection are defined by the formula
\[
C_a^{\on{KZ}}(z;\ka) := \left(\nabla_a^{\on{KZ}}\right)^p .
\]
They define an endomorphism of the KZ connection,
$\big[C_a^{\on{KZ}}, \nabla_b^{\on{KZ}}\big]=0$.

An operator $C^{\on{KZ}}_a(z;\ka)$ is a rational function in $z$ with the denominator
\[
D^{\on{KZ}}_a(z;\ka) = \prod_{j\ne a} \big(z_a^p-z_j^p\big).
\]
It is convenient to introduce the normalized $p$-curvature operator by the formula
\[
\tilde C^{\on{KZ}}_a(z;\ka) = D^{\on{KZ}}_a(z;\ka) C_a^{\on{KZ}}(z;\ka) .
\]
The normalized $p$-curvature operator
is a homogeneous polynomial in $z$ of degree $(n-2) p$ if nonzero.

 \subsection{KZ equations as a limit of qKZ equations}
 \label{eqn limit sec}

Let $f(z_1,\dots, z_n)$ satisfy the qKZ equations,
\[
f(z_1,\dots, z_a-\ka,\dots, z_n) =
K_a(z ;\ka) f(z),\qquad a=1,\dots,n.
\]
Let $\al$ be a formal parameter.
Define
$g(w_1,\dots,w_n;\al) := f(w_1/\al, \dots, w_n/\al)$.
Then
\begin{equation}\label{K)}
g(w_1,\dots, w_a-\al\ka,\dots, w_n;\al) =
K_a(w/\al ;\ka) g(w;\al),\qquad a=1,\dots,n,
\end{equation}
where
\begin{gather*}
K_a(w/\al;\ka)
=
R^{(a,a-1)}(w_a-w_{a-1}-\al \ka;\al) \cdots R^{(a,1)}(w_a-w_1-\al\ka;\al)
\\
\phantom{K_a(w/\al;\ka)
=}{}\times
 R^{(a,n)}(w_a-w_n;\al) \cdots R^{(a,a+1)}(w_a\<-w_{a+1};\al),
\\
R(u;\al)
=
\frac{u-\al P}{u-\al} = 1 - \al \frac{P-1}{u-\al} .
\end{gather*}
Equation \eqref{K} gives
\[
g - \al\ka \frac{\der g}{\der w_a} +\mc O\big(\al^2\big)
= \big(1-\al H_a(w) + \mc O\big(\al^2\big)\big) g.
\]
In the limit $\al\to 0$, we obtain the KZ differential equations
\[
\ka \frac{\der g}{\der w_a}(w) = H_a(w)g(w) .
\]

\begin{lem}
\label{lem C li C}

Let $\hat C_a(z;\ka)$ be a
 reduced $p$-curvature operators of the qKZ connection.
 Then
\[
C_a^{\on{KZ}}(z;\ka) =
\lim_{\al\to 0} \hat C_a(z_1/\al,\dots,z_n/\al;\ka)/\al^p.
\]

\end{lem}

\begin{proof} Let $C_{a,\al}$ be the $p$-curvature operator of equation \eqref{K)}.
Clearly,
\[
C_{a,\al}(w;\ka) = C_{a}(w_1/\al,\dots, w_n/\al;\ka).
\]
Hence we need to prove that
\begin{gather}
\label{CiC}
C_a^{\on{KZ}}(w;\ka)
= \lim_{\al\to 0}
(C_{a,\al}(w;\ka) -1)/\al^p.
\end{gather}
Let $\nabla_{a,\al}$ be the connection operator of the discrete connection \eqref{K)},
\[
(\nabla_{a,\al}g)(w) = K_a(w_1/\al, \dots, (w_a+\al\ka)/\al, \dots, w_n/\al;\ka) g(w_1, \dots, w_a+\al\ka, \dots, w_n).
\]
Then
\[
(\nabla_{a,\al}g)(w) = g - \al H_a(w) g + \al\ka \frac{\der g}{\der w_a} + \mc O\big(\al^2\big).
\]
Hence
\[
(\nabla_{a,\al} -1)g = \al \nabla^{\on{KZ}}_a g +\mc O\big(\al^2\big).
\]
Thus
\[
C_{a,\al}(w;\ka) -1 = (\nabla_{a,\al})^p-1 = (\nabla_{a,\al}-1)^p = \al^p \big(\nabla^{\on{KZ}}_a\big)^p +\mc O\big(\al^{p+1}\big).
\tag*{\qed}
\]\renewcommand{\qed}{}
\end{proof}

\begin{cor}
\label{cor deg C}
The degree of the reduced $p$-curvature operator $\tilde C_a(z;\ka)$ as a polynomial in $z$ is not greater than $(n-2)p$.
Moreover, we have
$\tilde C_a(z;\ka) = \tilde C_a^{\on{KZ}}(z;\ka) + f_1(z)$,
where $f_1(z)$ is a~polynomial in $z$ of degree less than $(n-2)p$.

\end{cor}

\begin{proof}
The first statement of the corollary
follows from the fact that the denominator of the reduced $p$-curvature operator $\hat C_a(z;\ka)$ is of degree $(n-1)p$.
Hence the numerator is of degree not greater than $(n-2)p$ for the limit in \eqref{CiC} to exist.
The second statement is clear.
\end{proof}

\subsection{Limit of solutions}

Let $f(z_1,\dots, z_n)$ be a polynomial solution of the qKZ equations \eqref{Ki}. Let $\deg f(z) = d$ and~${
f(z) = f_0(z) + f_1(z)}$,
where $ f_0(z)$ is a homogeneous polynomial of degree $d$ and $f_1(z)$ is a polynomial of degree less than $d$.

\begin{lem}
\label{thm thp}
The polynomial $f_0(z)$ is a solution of the KZ equations \eqref{KZ}.
\end{lem}

\begin{proof}
Let $\al$ be a formal parameter. Define
$g(w_1,\dots,w_n;\al) = \al^df(w_1/\al, \dots, w_n/\al)$.
Then
$
g(w_1,\dots,w_n;\al) = f_0(w) + \al g_1(w;\al)$,
where $g_1(w;\al) = \al^{d-1}f_1(w/\al)$ is a polynomial in $w$ and~$\al$.
We have
\[
g(w_1,\dots, w_a-\al\ka,\dots, w_n;\al) =
K_a(w/\al ;\ka) g(w;\al),\qquad a=1,\dots,n.
\]
Let
\[
d_a(w;\al) =
\prod_{j=1}^{a-1}(w_a-w_j -\al\ka-\al)
\prod_{j=a+1}^n(w_a-w_j -\al)
\]
be the denominator of $K_a(w/\al;\ka) $ and
\[
n_a(w;\al)
=
 \big(w_a-w_{a-1}-\al\ka - \al P^{(a,a-1)}\big) \cdots \big(w_a-w_{a+1}- \al P^{(a,a+1)}\big)
\]
the numerator.
Then \eqref{K)} can be written as a polynomial equation
\begin{gather}
\label{Dg}
d_a(w;\al) g(w_1,\dots, w_a-\al\ka,\dots, w_n;\al) =
n_a(w;\al) g(w;\al)
\end{gather}
in the variables $w$ and $\al$. Equation \eqref{Dg} gives an equation in $w$ for every fixed power of~$\al$ in~\eqref{Dg}.
The equation corresponding to the
first power of $\al$ after division by
$d_a(w;0)$ becomes
\[
\ka \frac{\der f_0}{\der w_a}(w) = \sum_{j\ne a} \frac{P^{(a,j)}-1}{w_a-w_j} f_0(w) .\tag*{\qed}
\]\renewcommand{\qed}{}
\end{proof}

\subsection{Dual qKZ equations}

Denote $W=L^{\ox n}$. Let $W^*$ be the dual space of $W$, and $\langle, \rangle \colon W^*\ox W\to\K$ the canonical pairing.
Let $K_a^*(z;\ka) \colon W^* \to W^*$ be the operators dual to the operators $K_a(z;\ka)$. Denote
$\tilde K_a(z;\ka) = (K_a^*(z;\ka))^{-1}$.

We have
\[
\tilde K_a(z_1,\dots, z_b-\ka,\dots, z_n;\ka) \tilde K_b(z;\ka)
= \tilde K_b(z_1,\dots, z_a-\ka,\dots, z_n;\ka) \tilde K_a(z;\ka)
\]
for all $a$, $b$, and also
\begin{gather}
\label{per}
\langle x, y\rangle = \big\langle \tilde K_a(z;\ka) x, K_a(z;\ka) y\big\rangle
\end{gather}
for all $a$ and $x\in W^*, y\in W$.

The system of difference equations with step $\ka$,
\begin{equation}
\label{Kid}
\tilde s(z_1,\dots, z_a-\ka,\dots, z_n) =
\tilde K_a(z ;\ka) \tilde s(z),\qquad a=1,\dots,n,
\end{equation}
for an $W^*$-valued
function $\tilde s(z)$, is called the dual qKZ equations.

If $s(z)$ is a solution of the qKZ equations \eqref{Ki} and
$\tilde s(z)$ is a solution of the dual qKZ equations~\eqref{Kid}, then
the function $\langle \tilde s(z), s(z) \rangle$ is $\ka$-periodic with respect to
every $z_a$,
 \[
-
\langle \tilde s(z_1,\dots,z_a-\ka, \dots, z_n),
s(z_1,\dots,z_a-\ka, \dots, z_n, )\rangle
=
\langle \tilde s(z),
s(z)\rangle .
\]

The set of vectors $\{v_I\mid I \subset\{1,\dots,n\}\}$ is a basis of $W$. Define
the nondegenerate symmetric bilinear form $S$ on $W$ by the formula
$S(v_I,v_J) = \delta_{I,J} $.
The form is called the (tensor) Shapovalov form. We identify $W^*$ and $W$ with the help
of the Shapovalov form.

Under this identification, the operators $R^{(i,j)}(u)$ become symmetric. Using \eqref{unit}
we obtain the following formula
for the operators $\tilde K_a(z;\ka)$ as operators on $W$,
\begin{align*}
&\tilde K_a(z;\ka) =
R^{(a,a-1)}(-z_a+z_{a-1}+\ka) \cdots R^{(a,1)}(-z_a+z_1+\ka)
\\
&\hphantom{\tilde K_a(z;\ka) =}{} \times
R^{(a,n)}(-z_a+z_n) \cdots R^{(a,a+1)}(-z_a\<+z_{a+1}) .
\end{align*}
In other words,
\begin{gather*}
\tilde K_a(z_1,\dots,z_n;\ka)
=
 K_a(-z_1,\dots,-z_n;-\ka),\nonumber
\\
S(K_a(-z_1,\dots,-z_n;-\ka) x, K_a(z_1,\dots,z_n;\ka)y)
= S(x,y)\label{pm}
\end{gather*}
for all $x,y\in W$.

Now the dual qKZ equations for a $W$-valued function $\tilde s(z)$ take the form
\[
\tilde s(z_1,\dots, z_a-\ka,\dots, z_n) =
K_a(-z_1,\dots,-z_n;-\ka) s(z),\qquad a=1,\dots,n.
\]

These formulas prove the following lemma.

\begin{lem}\label{lem inv}
Let $s(z)$ be a solution of the qKZ equations \eqref{Ki} with step $\ka$ and
$\tilde s(z)$ a solution of the qKZ equations \eqref{Ki} with step $-\ka$. Then
the function $S(\tilde s(-z),s(z))$ is $\ka$-periodic with respect to every variable $z_a$,
\[
S(\tilde s(-z_1,\dots, - (z_a-\ka), \dots, -z_n),
s(z_1,\dots,z_a-\ka, \dots, z_n ))
=
S(\tilde s(-z),
s(z)) .
\]
\end{lem}

See also Corollary \ref{cor ort}.

\section{Pochhammer polynomials}\label{sec PP}

For $\ka\in\K$ and $m$ a positive integer, define Pochhammer polynomial $(t;\ka)_m\in \K[t]$
by the formula~${
(t;\ka)_m=\prod_{i=1}^{m}(t-(i-1)\ka)}$.\footnote{See \url{https://en.wikipedia.org/wiki/Falling_and_rising_factorials}.}
We have
\begin{gather}
(t-\ka;\ka)_m = (t;\ka)_m \frac{t-\ka m}t,
\qquad
(t+\ka ;\ka)_m =(t;\ka)_m \frac{t+\ka }{t-(m-1)\ka }.\nonumber
\\
(t+z;\ka)_m
=
\sum_{i=0}^m {m\choose i} (t;\ka)_i(z;\ka)_{m-i},\nonumber
\\
(t;\ka)_i(t;\ka)_j
=
 \sum_{l=0}^{\min(i,j)}
\binom{i}{l} \binom{j}{l} l! \ka^l (t;\ka)_{i+j-l},\label{prod}
\\
(t;\ka)_m
=
\sum _{l=0}^{m}s_1(m,l) \ka^{m-l} t^{l},\nonumber
\\
t^m
=
\sum _{l=0}^{m}s_2(m,l) \ka^{m-l} (t;\ka)_l,\label{PS2}
\end{gather}
where the integers $s_1(m,l)$ and $s_2(m,l)$ are Stirling numbers of the first and second kind,
 respectively.
 Notice that
\begin{gather} \label{s=1}
 s_1(m,m)=s_2(m,m) =1.
 \end{gather}
We also have $ (t;\ka)_p = t^p-\ka^{p-1}t$,
\[
(t+z;\ka)_p = (t+z)^p-\ka^{p-1}(t+z) =t^p-\ka^{p-1}t+z^p-\ka^{p-1}z = (t;\ka)_p+(z;\ka)_p .
\]

We call a polynomial $f(t)\in\F_p[t]$ a quasi-constant
if $f(t-\kappa) = f(t)$. The quasi-constants are polynomials in $t^p-\ka^{p-1}t$.
A Pochhammer polynomial $(t;\ka)_m$ is a quasi-constant if $m$ is divisible by
$p$. Then $(t;\ka)_{pa} = \big(t^p-\ka^{p-1}t\big)^a$.

Let $A$ be a $\K$-algebra, for example, $A=\K[z_1,\dots,z_n]$.
The polynomials $\{(t;\ka)_m\mid m\geq 0\}$ form an $A$-basis of the ring of polynomials $A[t]$.

\section[p-hypergeometric solutions for k in F\_p\^times subset K]{$\boldsymbol{ p}$-hypergeometric solutions for $\boldsymbol{\ka\in\F_p^\times \subset \K}$}

\subsection[Solutions in Sing L\^{}\{ox n\}n-2]{Solutions in $\boldsymbol{\Sing L^{\ox n}[n-2]}$}

 In this paper, we study solutions of the qKZ equations with values in
 $V:=\Sing L^{\ox n}[n-2]$.
 The space $ L^{\ox n}[n-2]$
 has a basis
$
v^{(i)} = v_1\ox \dots\ox v_2\ox\dots \ox v_1$, $
i=1,\dots,n$,
where the only $v_2$ stays at the $i$-th place. In this basis, the subspace
$V$ consists of all vectors with the sum of coordinates equal to zero.
We identify $L^{\ox n}[n-2]$ with $\K^{n}$ and the subspace $V$ with
the vector space
\begin{gather}\label{s0}
\{ x\in \K^{n}\mid x_1+\dots + x_n=0\}.
\end{gather}

\subsection{Master polynomial and weight functions}

For $\ka\in\F_p^\times$, let $0<k<p$ be the positive integer such that
\begin{gather}\label{C}
\ka k \equiv -1 \pmod{p}.
\end{gather}
Define master polynomial
$\Phi(t,z;\ka) = \prod_{a=1}^n (t-z_a;\ka)_k$,
where $k$ is defined in~\eqref{C}.
For $1\leq a\leq n$, define the weight functions
\[
\eta_a(t,z)
=
\frac {1}{t-z_{a}} \prod_{j=1}^{a-1}\frac {t-z_{j}+1}{t-z_{j}}
,
\qquad
Q_a(t,z;\ka)
=
\Phi(t,z;\ka) \eta_a (t,z).
\]
Notice that the formula for the function $\eta_a(t,z)$ does not depend on $p$ and $\ka$.

\begin{lem}
The function $Q_a (t,z;\ka)$
 is a product of $n$ Pochhammer polynomials:
\[
 Q_a (t,z;\ka) =
 \left(\prod_{j=1}^{a-1} (t-z_j - \ka;\ka)_k\right)
 (t-z_a - \ka;\ka)_{k-1}
 \left(\prod_{j=a+1}^{n} (t-z_j;\ka)_k\right).
 \]
 \end{lem}

Define a vector of polynomials
\begin{gather}\label{U}
Q(t,z;\ka) = (Q_1(t,z;\ka), \dots, Q_n(t,z;\ka))^\intercal =\sum_i Q^i(z;\ka) (t;\ka)_i,
\end{gather}
where $M^\intercal$ denotes the transpose matrix of a matrix $M$ and
$
 Q^i(z;\ka) = \big(Q^i_1(z;\ka), \dots, Q^i_n(z;\ka)\big)^\intercal
$
 are vectors of polynomials in $z$.

\begin{exmp} For $n=2$,
\[
Q(t,z;\ka) =\big( (t_1-z_1 -\ka;\ka)_{k-1} (t_1-z_2;\ka)_{k},
(t_1-z_1 -\ka;\ka)_{k} (t_1-z_2 -\ka;\ka)_{k-1}\big)^\intercal.
\]
\end{exmp}

\begin{thm}[{\cite[Theorem 5.1]{MV}}]\label{thm MV}
For any positive integer $\ell$, the vector $Q^{\ell p-1}(z;\ka)$
 of polynomials in $z$ is a solutions of the qKZ equations with step $\ka$
 and values in
 $V=\Sing L^{\ox n}[n-2]$,
in particular,
$
\sum_{a=1}^nQ^{\ell p-1}_a(z;\ka)=0$,
 see \eqref{s0}.
\end{thm}

Let $[x]$ denote the integer part of a real number $x$.
The vector $Q^{\ell p-1}(z;\ka)$ is zero if $\ell \not\in \big\{1, \dots, \big[\frac{n{k}}p\big]\big\}$ for degree reasons.
The vectors of polynomials
$Q^{\ell p-1}(z;\ka)$, $
\ell = 1, \dots, \big[\frac{n{k}}p\big]$,
 are~called the $p$-hypergeometric solutions of
the qKZ equations with step $\ka$
 and values in $V$.
 Denote~${d(\ka) = \big[\frac{n{k}}p\big]}$.

Notice that if $\frac{n{k}}p < 1$, then there are no $p$-hypergeometric solutions.

\subsection[Step -ka]{Step $\boldsymbol{-\ka}$}

The integer $k$ satisfies the inequalities $0< k <p$ and the congruence
$\ka k \equiv -1 \pmod{p}$, see \eqref{C}.
Then the integer $ p-k$ satisfies the inequalities
$0< p-k<p$ and the congruence
$-\ka (p- k) \equiv -1 \pmod{p}$.
Hence
$
\Phi(t,z;-\ka)
= \prod_{a=1}^n (t-z_a;-\ka)_{p-k} $.
Recall that
\begin{gather*}
Q_a(t,z;-\ka)
=
\Phi(t,z;-\ka) \eta_a (t,z),
\\
Q(t,z;-\ka) = (Q_1(t,z;-\ka), \dots, Q_n(t,z;-\ka))^\intercal =\sum_i Q^i(z;-\ka) (t;-\ka)_i .
\end{gather*}

\begin{cor}
\label{cor -ka}
The vectors
\[Q^{\ell p-1}(z;-\ka), \qquad \ell = 1, \dots, \left[\frac{n(p-k)}p\right],
\]
 are solutions of the qKZ equations with step $-\ka$
 and values in $V$.
\end{cor}

If $p$ does not divide $n$, then the total number of $p$-hypergeometric solutions of the qKZ equations with values in
$V$ and steps $\ka$ and $-\ka$ equals
\[
\left[\frac{n{k}}p\right] + \left[\frac{n(p-{k})}p\right] = n-1 = \dim V.
\]

\subsection[p-hypergeometric solutions of KZ equations]{$\boldsymbol{ p}$-hypergeometric solutions of KZ equations}

In this subsection, we remind the construction in \cite{SV} of polynomial solutions modulo
 $p$ of the
 differential KZ equations with values in $V=\Sing L^{\ox n}[n-2]$.

Let $0< k <p$ be the positive integers such that $\ka k \equiv -1 \pmod{p}$.
Define master polynomial~${
\bar\Phi(t,z, \ka) = \prod_{a=1}^n (t-z_a)^k }$.
For $1\leq a\leq n$, define the weight functions
\[
\bar w_a(t,z) = \frac {1}{t-z_{a}},\qquad
\bar Q_a(t,z;\ka) = \bar\Phi(t,z;\ka) \bar w_a(t,z).
\]
Then $\bar Q_a(t,z;\ka)$ is a polynomial in $t,z$.
Define a vector of polynomials in $t$, $z$,
\[
\bar Q(t,z;\ka) = \big(\bar Q_1(t,z;\ka), \dots, \bar Q_n(t,z;\ka)\big)^\intercal =\sum_i \bar Q^i(z;\ka) t^i,
\]
where
$
 \bar Q^i(z;\ka) = \big(\bar Q^i_1(z;\ka), \dots, \bar Q^i_n(z;\ka)\big)^\intercal
$
 are vectors of polynomials in $z$.

\begin{thm}[\cite{SV}]\label{SV2 thm}
For any positive integer $\ell$, the vector
 $\bar Q^{\ell p-1}(z;\ka)$
 is a solution
 of the differential KZ equations with parameter $\ka$
 and values in $V=\Sing L^{\ox n}[n-2]$.
\end{thm}

The vector $\bar Q^{\ell p-1}(z;\ka)$ is zero if $\ell \not\in \big\{1, \dots, \big[\frac{n{k}}p\big]\big\}$ for degree reasons.
The vectors
\[
\bar Q^{\ell p-1}(z,\ka), \qquad
\ell = 1, \dots, \left[\frac{n{k}}p\right],
\]
 are called the $p$-hypergeometric solutions of
the differential KZ equations with parameter $\ka$
 and values in $V$.

Notice that $\bar Q_a(t,z;\ka)$ are homogeneous polynomials in variables $t$, $z$ of degree
$n k-1$, and~${\bar Q^{\ell p-1}(z;\ka)}$ are vectors of homogeneous polynomials in $z$ of degree
$nk - \ell p$.

\subsection[Top-degree part of p-hypergeometric solutions]{Top-degree part of $\boldsymbol{ p}$-hypergeometric solutions}

It turns out that the top-degree part of a $p$-hypergeometric solution
$Q^{\ell p-1}(z;\ka)$ of the qKZ equations is the $p$-hypergeometric solution
 $\bar Q^{\ell p-1}(z;\ka)$ of the KZ equations.

We start with an abstract lemma. Let $t, z_1,\dots,z_n, \al$ be variables.
Given a homogeneous polynomial $P(t, z_1,\dots,z_n,\al)$ of degree $d$ in the variables $t$, $z$, $\al$
and an integer $e$, $0\leq e\leq d$, we construct two polynomials
$\bar P_e(z_1,\dots,z_n)$ and $P_e(z_1,\dots,z_n,0)$ as follows.

On the one hand, we have
\begin{gather}
\label{Pta}
P(t, z_1,\dots,z_n,\al)= \sum_{d_0+\dots+d_{n+1}=d} a_{d_0,\dots,d_{n+1}} t^{d_0} z_1^{d_1}\cdots z_n^{d_n}\al^{d_{n+1}} .
\end{gather}
Then
\begin{gather}
\label{Ptae}
P(t, z_1,\dots,z_n, 0)= \sum_{d_0+\dots+d_{n}=d} a_{d_0,\dots,d_n, 0} t^{d_0} z_1^{d_1}\cdots z_n^{d_n} .
\end{gather}
Denote
\[
 \bar P_e(z_1,\dots,z_n)= \sum_{e+d_1+\dots+d_{n}=d} a_{ e,d_1, \dots,d_n, 0} z_1^{d_1}\cdots z_n^{d_n},
\]
the coefficient of $t^e$ in \eqref{Ptae}.

On the other hand, applying formula \eqref{PS2} to each $t^{d_0}$ at \eqref{Pta}, we rewrite this sum as
\begin{gather}
\label{PPe}
P(t, z_1,\dots,z_n,\al)=\sum_{d_0+\dots+d_{n+1}=d} b_{d_0,\dots,d_{n+1}} (t;\al)_{d_0} z_1^{d_1}\cdots z_n^{d_n}\al^{d_{n+1}} .
\end{gather}
Denote
\[
P_e(z_1,\dots,z_n,\al)=\sum_{e+d_1+\dots+d_{n+1}=d} b_{e,d_1,\dots,d_{n+1}} z_1^{d_1}\cdots z_n^{d_n}\al^{d_{n+1}},
\]
the coefficient of $(t;\al)_e$ in \eqref{PPe}. Then
\[
P_e(z_1,\dots,z_n,0)=\sum_{e+d_1+\dots+d_{n}=d} b_{e,d_1,\dots,d_{n},0} z_1^{d_1}\cdots z_n^{d_n} .
\]

\begin{lem}\label{lem top deg}
We have $P_e(z_1,\dots,z_n,0) = \bar P_e(z_1,\dots,z_n)$.
\end{lem}

\begin{proof}By formula \eqref{s=1}, we have
$a_{e,d_1,\dots,d_n, 0} = b_{e,d_1,\dots,d_n, 0}
$
for any $e,d_1,\dots,d_n$. This implies the lemma.
\end{proof}

\begin{cor}[\cite{MV}]\label{cor q-kz}
Let $Q^{\ell p-1}(z;\ka)$ be a $p$-hypergeometric solution of the qKZ equations from Theorem~{\rm \ref{thm MV}}. Then
$
Q^{\ell p-1}(z;\ka) =
\bar Q^{\ell p-1}(z;\ka) +\cdots$,
where $\bar Q^{\ell p-1}(z;\ka)$ is the corresponding
$p$-hypergeometric solution of the KZ equations from Theorem~{\rm \ref{SV2 thm}}, and
the dots denote the
terms of degree less than $nk - \ell p=\deg \bar Q^{\ell p-1}(z;\ka)$.
\end{cor}

\begin{proof}
The corollary is proved by application of Lemma \ref{lem top deg} to
the polynomial $Q(t,z;\ka)$ and integer $e=\ell p-1$. Then
it is easy to see that
$\bar P_{\ell p-1}(z_1,\dots,z_n) = \bar Q^{\ell p-1}(z;\ka) $ and
$P_{\ell p-1}(z_1,\dots,z_n, 0)$ is the top-degree part of $Q^{\ell p-1}(z;\ka)$.
\end{proof}

\section{Linear independence}

\subsection{Lexicographical order}
Define length-lexicographical order on monomials
$v=
z_1^{d_1}\cdots z_n^{d_n}$: $v<w$ if degree$(v)<$degree$(w)$ or if
degree$(v)=$degree$(w)$ and $v$ is smaller than $w$ in lexicographical order.

For a polynomial
\smash{$f(z)=\sum_{d_1,\dots,d_n} a_{d_1,\dots,d_n} z_1^{d_1}\cdots z_n^{d_n}$} denote by $\text{L}f(z)$ the nonzero summand
$a_{d_1,\dots,d_n} z_1^{d_1}\cdots z_n^{d_n}$ with
the lexicographically largest monomial
$z_1^{d_1}\cdots z_n^{d_n}$.
 We call $\text{L}f(z)$
 the leading term of $f(z)$.

\subsection{Leading terms}\label{sec lt}

For $\ell \in \{ 1, \dots, d(\ka)\}$,
let $r(\ell)$ be the unique non-negative integer such that
\[
r(\ell)k \leq nk-\ell p < (r(\ell) + 1)k.
\]
We have $r(1) > r(2)> \cdots$. Denote
\[
g
a=(n-r(\ell))k - \ell p,
\qquad
u_\ell = \frac{(-1)^{nk-\ell p}}{k} \binom{k}{a}
(0, \dots,0, k-a, k, \dots, k)^\intercal \in \F_p^n,
\]
where $k-a$ is the $r(\ell)+1$-st coordinate. Notice that the integers $k$ and $k-a$ are not divisible by $p$, and
$u_\ell \in \F_p^n$ is a singular vector since the sum of its coordinates equals zero.

\begin{lem}[{\cite[Lemma 3.1]{V}}]\label{lem lcKZ}
Assume that $p$ does not divide $n$ and $\ell \in \{ 1, \dots,d(\ka) \}$. Then the leading term of
$\bar Q^{\ell p-1}(z;\ka)$ is given by the formula
\[
\on{L}\bar Q^{\ell p-1}(z;\ka)=(z_1\cdots z_{r(\ell)})^kz_{r(\ell)+1}^a u_\ell .
\]
\end{lem}

\begin{lem}\label{lem lead}
Assume that $p$ does not divide $n$ and $\ell \in\{ 1, \dots,d(\ka)\}$.
Let $Q^{\ell p-1}(z;\ka)$ be the corresponding $p$-hypergeometric solution of the qKZ equations and $\bar Q^{\ell p-1}(z;\ka)$
the corresponding $p$-hypergeometric solution of
the KZ equations. Then their leading terms are equal,
\[
\on{L} Q^{\ell p-1}(z;\ka) = \on{L} \bar Q^{\ell p-1}(z;\ka).
\]
\end{lem}

\begin{proof}
The lemma follows from Corollary \ref{cor q-kz}.
\end{proof}

\begin{cor}If $p$ does not divide $n$ and $\ell \in\{ 1, \dots,d(\ka)\}$, then
the leading term of
$ Q^{\ell p-1}(z;\ka)$ is given by the formula
\begin{gather}\label{L bar}
\on{L} Q^{\ell p-1}(z;\ka)=(z_1\cdots z_{r(\ell)})^kz_{r(\ell)+1}^a u_\ell .
\end{gather}
\end{cor}

\begin{exmp}\label{ex 5.4}
Let $n=3$ and $d(\ka)=\big[\frac{3k}p\big]=1$. Then
$\frac p3 < k< \frac{2p}3 $.
In this case there is exactly one $p$-hypergeometric solution $Q^{p-1}(z;\ka)$.
If $\frac p2< k<\frac{2p}3$, then the leading term of $Q^{p-1}(z;\ka)$ is
\[
z_1^kz_2^{2k-p}\frac{(-1)^{3k-p}}{k}\binom{k}{2k-p} (0,p-k,k )^\intercal .
\]
If $\frac p3< k<\frac{p}2$, then the leading term of $Q^{p-1}(z;\ka)$ is
\[
z_1^{3k-p}
\frac{(-1)^{k}}{k}\binom{k}{3k-p} (p-2k,k,k )^\intercal .
\]
\end{exmp}

Consider the collections of the $n$-vectors $Q^{\ell p-1}(z;\ka)$,
$\ell = 1, \dots, d(\ka)$, as an $(n\times d(\ka))$-matrix. For
$I=\{1\leq i_1<\dots <i_{d(\ka)}\leq n\}$,
 denote by $M_I(z, \ka)$ the $(d(\ka)\times d(\ka))$-minor
of that matrix located at the rows with indices in $I$.

\begin{lem}\label{lem minor}
If $p$ does not divide $n$ and
 $I=\{r(d(\ka)) <\dots< r(1)\}$, then the minor $M_I(z;\ka)$ is a nonzero polynomial.
\end{lem}

\begin{proof}
The lemma is a corollary of formula \eqref{L bar}.
\end{proof}

\begin{thm}\label{thm ind}
If $p$ does not divide $n$, then the $p$-hypergeometric solutions
$Q^{\ell p-1}(z;\ka)$,
$\ell = 1, \dots, d(\ka)$,
of the qKZ equations are linearly independent over the field $\K(z)$.
\end{thm}

\begin{proof} The statement follows from formula \eqref{L bar} and the fact that the vectors
$u_\ell$ are linearly independent over $\K$.
\end{proof}

\section{Orthogonality relations}\label{}

\subsection{Statement}

Recall that $d(-\ka) = \big[\frac{n(p-k)}p\big]$.

\begin{thm}\label{thm orth}Let $p>n$
and $0< d(\ka) <n-1$.
 Then for any $\ell \in \{1, \dots,d(\ka)\}$ and
$m \in \{1, \dots, d(-\ka)\}$,
we have
\begin{gather}
\label{ortA}
S\big(Q^{m p-1}(-z;-\ka), Q^{\ell p-1}(z;\ka)\big)
=
\sum_{a=1}^n Q^{m p-1}_a(-z;-\ka) Q^{\ell p-1}_a(z;\ka)
 = 0,
\end{gather}
where $S$ is the Shapovalov form.
\end{thm}

The theorem is proved in Sections \ref{sec 6.2} and \ref{sec proof}.

\begin{rem}
It is easy to see that the Shapovalov form on $V$ is nondegenerate if $p$ does not divide~$n$. Indeed, the vectors
$e_1=(1,-1,0,\dots,0)$,
$e_2=(0, 1,-1,0,\dots,0)$, \dots, $e_{n-1}=(0, \dots, 0, 1,-1)$ form a basis of $V$, and the determinant of the Shapovalov form in this basis equals $n$.
\end{rem}

\begin{rem}Formula \eqref{ortA} and Corollary
\ref{cor q-kz} imply the orthogonality relations for the $p$-hypergeo\-metric solutions of the KZ equations,
\begin{gather}
\label{ortAkz}
S\big(\bar Q^{m p-1}(-z;-\ka), \bar Q^{\ell p-1}(z;\ka)\big)
=
\sum_{a=1}^n \bar Q^{m p-1}_a(-z;-\ka) \bar Q^{\ell p-1}_a(z;\ka).
 = 0.
\end{gather}
Two different proofs of formula \eqref{ortAkz} are given in
\cite[Theorem~3.11]{VV1}
and
\cite[Appendix~A]{VV1}.
\end{rem}

\subsection{Special restrictions}
\label{sec 6.2}

Let $I\subset \{1,\dots,n\}$ be a nonempty subset,
 $I = \{1\leq i_1<i_2<\dots<i_a\leq n\}$.
Denote $S_I$ the system of equations
\begin{gather}
 \label{S_I}
 z_{i_b} = ((b-1)k -1)\ka,
 \qquad
 b=1,\dots,a.
\end{gather}
 For a polynomial $f(z)$, define $f(z)_{S_I}$ to be the polynomial $f(z)$ in which the variables
 $(z_i)_{i\in I}$ are replaced by multiples of $\ka$ according to formulas \eqref{S_I}.

\begin{lem}\label{lem S_I}
Let $Q(t,z;\ka)$ be the vector of polynomials defined by \eqref{U}.
Let $I = \{1\leq i_1<i_2<\dots<i_a\leq n\}\subset\{1,\dots,n\}$
be a nonempty subset. Then
\[
Q(t,z;\ka)_{S_I} = (t;\ka)_{ak-1} (P_1(t,z), \dots, P_n(t,z))^\intercal,
\]
where $P_1(t,z), \dots, P_n(t,z)$ are suitable polynomials of degree $(n-a)k$.
\end{lem}

\begin{proof}
The proof is straightforward.
\end{proof}

\begin{cor}\label{cor restr}
Let $Q^{\ell p-1}(z;\ka)$ be a $p$-hypergeometric solution and $\ell p < ak$. Then
\begin{gather}
\label{sol res}
Q^{\ell p-1}(z;\ka)_{S_I} =0.
\end{gather}
\end{cor}

\begin{proof}
 For $j=1,\dots,n$, the $j$-th coordinate of $Q(t,z;\ka)_{S_I}$ equals
$(t;\ka)_{ak-1} P_j(t,z)$. Using~\eqref{prod}, we rewrite this as
$\sum_{i\geq ak-1} c_i(z) (t;\ka)_i$ for suitable $c_i(z)$.
We observe that $(t;\ka)_{\ell p-1}$ does not enter this sum.
 This proves the corollary.
\end{proof}

\subsection{Proof of Theorem \ref{thm orth}}\label{sec proof}

Denote
\[
G_{\ell,m}(z;\ka) = S\big(Q^{m p-1}(-z;-\ka), Q^{\ell p-1}(z;\ka)\big).
\]
 Then
\begin{align*}
G_{\ell,m}(-z;\ka)
={}&
 S\big(Q^{m p-1}(z;-\ka), Q^{\ell p-1}(-z;\ka)\big)\\
={}&
S\big(Q^{\ell p-1}(-z;\ka), Q^{m p-1}(z;-\ka)\big) = G_{m,\ell}(z;-\ka).
\end{align*}
Hence, $G_{\ell,m}(z;\ka)=0$
for all $\ell$, $m$ if and only if
$G_{m,\ell}(z;-\ka)=0$.

By Lemma \ref{lem inv},
the function $G_{\ell,m}(z;\ka)$ is a polynomial in
$\F_p[z_1^p-z_1,\dots, z_n^p-z_n]$.
Denote $h(x) = x^p-x$.

\begin{lem}\label{lem deg}
Given $\ell$ and $m$, we have
\[
G_{\ell,m}(z;\ka) = c_0 +
\sum_{b=1}^{n-\ell-m}
\sum_{1\leq j_1<\dots < j_b\leq n} c_{j_1,\dots,j_b}
h(z_{j_1})\cdots h(z_{j_b}), \qquad c_0, c_{j_1,\dots,j_b}\in
\F_p .
\]
\end{lem}

\begin{proof}
On the one hand, we have
\[
\deg Q^{\ell p-1}(z;\ka) = kn-\ell p\qquad \text{and}\qquad
\deg Q^{m p-1}(-z;-\ka) = (p-k)n-m p.
\]
 Hence
$\deg G_{\ell,m}(z;\ka) \leq (n-\ell-m)p$.

On the other hand, for any $j=1,\dots, n$, we have $\deg_{z_j}
G_{\ell,m}(z;\ka)\leq p$
since
$\deg_{z_j} Q(t,z;\ka) = k$
and
$\deg_{z_j} Q(t,-z;-\ka) = p-k$.
These two remarks prove the lemma.
\end{proof}

The $p$-hypergeometric solution $Q^{\ell p-1}(z, \ka)$ is associated with the
 integer
$0<k<p$ such that~${\ka k\equiv -1\pmod{p}}$, while
the $p$-hypergeometric solution $Q^{m p-1}(z, -\ka)$ is associated with the integer
$0<p-k<p$ such that $-\ka (p-k)\equiv -1\pmod{p}$.
 Given $m$, $\ell$, we say that $k$ is a~good parameter if
$\ell(p-k)< mk$.
We say that $p-k$ is a good parameter if $mk < \ell(p-k)$.
Notice that $ mk\ne \ell(p-k)$. Otherwise $p$ must divide $ \ell+m$ which is impossible since
$\ell + m\leq n-1<p$.

Having the two integers $k$, $p-k$ and two solutions
$Q^{\ell p-1}(z;\ka)$ and $Q^{m p-1}(z;-\ka)$, we may and will assume that $k$ denotes the good parameter.

\begin{lem}\label{lem fin}
Given $Q^{\ell p-1}(z;\ka)$ and $Q^{m p-1}(z;-\ka)$, assume that $k$ is a good parameter.
Then $G_{\ell,m}(z;\ka)= 0$.
\end{lem}

\begin{proof}
Let $\{1,\dots,n\}= I\cup J$ be a partition where
$|I|=\ell+m$ and $|J|=n-\ell-m$.
We may apply formula \eqref{sol res} to
$Q^{\ell p-1}(z;\ka)$ with $a=\ell + m$ since $k$ is a good parameter and hence
$\ell p < (\ell+m)k$. Hence we have $G_{\ell,m}(z;\ka)_I = 0$, where
\begin{gather}
\label{SQQI}
G_{\ell,m}(z;\ka)_I = c_0 +
\sum_{b=1}^{n-\ell-m}
\sum_{ \{j_1<\dots < j_b\}\subset J} c_{j_1,\dots,j_b}
h(z_{j_1})\dots h(z_{j_b}).
\end{gather}
Hence the right-hand side polynomial at \eqref{SQQI} is the zero polynomial for any subset
$I$ with $|I|=\ell + m$. Therefore, $G_{\ell,m}(z;\ka)$ is the zero polynomial.
\end{proof}

Lemma \ref{lem fin} implies Theorem \ref{thm orth}.

\section{Invariant subbundles}

\subsection{Subbundle of qKZ connection}

Assume that $p>n$ and $\ka\in\F_p^\times$.
Consider the discrete qKZ connection
on the direct product~${V\times \K^n \to \K^n}$, where
$
V=\{ x\in \K^{n}\mid x_1+\dots + x_n=0\}
$
and the qKZ operators are defined by formula \eqref{K},
\begin{align*}
&K_a(z;\ka) =
R^{(a,a-1)}(z_a-z_{a-1}-\ka) \cdots R^{(a,1)}(z_a-z_1-\ka)
\\
&\hphantom{K_a(z;\ka) =}{} \times
R^{(a,n)}(z_a-z_n) \cdots R^{(a,a+1)}(z_a\<-z_{a+1}) .
\end{align*}
The connection has singularities at the points of $\K^n$
where the qKZ operators have poles or are degenerate.

The $R$-matrix
$R(u) = \frac{u-P}{u-1}$ has a pole if $u=1$ and is degenerate if
$u=-1$.
Denote by $H_{i,j,m}$ the affine hyperplane in $\K^n$ defined by the equation
$z_i-z_j-m=0$ where $1\leq i<j\leq n$, $m\in\F_p$.
 Let $\bar{\mc A}^\circ$ be the arrangement in $\K^n$
of all hyperplanes $H_{i,j,m}$. Let
$\bar{\mc A}=\K^n-\bar{\mc A}^\circ$ denote its complement.
The qKZ operators are well-defined over $\bar{\mc A}$ and are nondegenerate.

Assume that $0<d(\ka)<n-1$.
Then the $p$-hypergeometric solutions
$Q^{\ell p-1}(z;\ka)$, $
\ell = 1, \dots, d(\ka)$,
 define flat sections of the qKZ connection which we call the $p$-hypergeometric sections.

 Recall the minors $M_I(z;\ka)$
 defined for any $I=\{1\leq i_1<\dots < i_{d(\ka)} \leq n\}$ in Section \ref{sec lt}
 with the help of these $p$-hypergeometric sections.
 Denote by $\mc A(\ka)$ the Zariski open subset of $\bar{\mc A}$ consisting of points
$b\in \bar{\mc A}$ such that at least one of the minors $M_I(z;\ka)$ is nonzero at $b$.

For any point $b \in \mc A(\ka)$, the vectors
$Q^{\ell p-1}(b,\ka)$, $\ell = 1, \dots, d(\ka)$, are linearly independent
and span a $d(\ka)$-dimensional $\K$-vector subspace
$\mc S(b,\ka)$ of $V$. These subspaces $\mc S(b,\ka)$, $b\in\mc A(\ka)$, form a
vector subbundle $\mc S(\ka) \to \mc A(\ka)$ of the trivial bundle
 $V\times \mc A(\ka)\to \mc A(\ka)$.

 \begin{rem}
 Notice that the minors $M_I(z;\ka)$ are polynomials in $z$ with coefficients in $\F_p$ and are independent of the field $\K$.
 Notice also that the base $\mc A(\ka)$ is invariant with respect to the affine translations,
$
 (z_1,\dots,z_n) \mapsto (z_1,\dots,z_a-\ka,\dots, z_n)$, $ a=1,\dots,n$.
 \end{rem}

 The subbundle $\mc S(\ka) \to \mc A(\ka)$ is invariant under the qKZ connection,
 and the $p$-hypergeometric sections form a flat basis of the space of its sections.

 We also consider the quotient bundle $\mc Q(\ka)\to \mc A(\ka)$ with fibers
$V/\mc S(b,\ka)$. The qKZ connection on $V\times \mc A(\ka)\to\mc A(\ka)$ induces
a discrete connection on $\mc Q(\ka)\to \mc A(\ka)$ which we also call the qKZ connection.
Notice that the rank of $\mc Q(\ka)\to \mc A(\ka)$ equals
\[
\dim V- \on{rank} \mc S(\ka)=n-1 - \left[\frac{n{k}}p\right] = \left[\frac{n{(p-k)}}p\right] = d(-\ka).
\]

 If $d(\ka)=0$, we define $\mc A(\ka) = \bar{\mc A}$. In this case, we define $\mc S(\ka) \to \mc A(\ka)$
 to be the rank 0 subbundle of $V \times \mc A(\ka)$ and also define $\mc Q(\ka)\to \mc A(\ka)$ to be
$V \times \mc A(\ka)$.

 \subsection{Subbundle of dual qKZ connection} \label{sec 7.2}

 Assume that $p>n$ and $\ka\in \F_p^\times$.
 Consider the dual discrete qKZ connection on $V\times \K^n \to \K^n$ defined by formulas \eqref{Kid} and \eqref{pm}.
 The dual qKZ operators $ K_a(-z_1,\dots,-z_n;-\ka)$, $a=1,\dots,n$,
 are well-defined over $\bar{\mc A}$ and are nondegenerate.

Assume that $0<d(-\ka)<n-1$. This assumption is equivalent to the assumption $0<d(\ka)<n-1$.
Consider the $p$-hypergeometric solutions
$Q^{m p-1}(z;-\ka)$,
$m = 1, \dots, d(-\ka)$, of the qKZ equations step $-\ka$. Then the $V$-valued polynomials
$
Q^{m p-1}(-z;-\ka)$, $
m = 1, \dots, d(-\ka)$,
 define flat sections of the dual qKZ connection which we also call the $p$-hypergeometric sections.

 Recall the minors $M_I(z;-\ka)$
 defined for any $I=\{1\leq i_1<\dots < i_{d(-\ka)} \leq n\}$ in Section~\ref{sec lt}
 with the help of the $p$-hypergeometric solutions $Q^{m p-1}(z;-\ka)$,
$m = 1, \dots, d(-\ka)$.
 Denote by~$\mc B(\ka)$ the Zariski open subset of $\bar{\mc A}$ consisting of points
$b\in \bar{\mc A}$ such that at least one of the polynomials $M_I(-z;-\ka)$ is nonzero at $b$.

For any point $b \in \mc B(\ka)$, the vectors
$Q^{m p-1}(-b,-\ka)$, $m = 1, \dots, d(-\ka)$,
are linearly inde\-pendent
and span a $d(-\ka)$-dimensional $\K$-vector sub\-space
$\mc S^*(b,-\ka)$ of space~$V$. These subspaces $\mc S^*(b,-\ka)$, $b\in\mc B(\ka)$, form a
vector subbundle $\mc S^*(\ka) \to \mc B(\ka)$ of the trivial bundle
 ${V\times \mc B(\ka)\to \mc B(\ka)}$.

The subbundle $\mc S^*(\ka) \to \mc B(\ka)$ is invariant with respect to the dual qKZ connection,
 and the $p$-hypergeometric sections
 $Q^{m p-1}(-z;-\ka)$, $m = 1, \dots, d(-\ka)$,
 form a flat basis of the space of its sections.

Consider the restriction of the bundles
$\mc S^*(\ka) \to \mc B(\ka)$ and
$\mc Q(\ka)\to \mc A(\ka)$ to
$\mc A(\ka) \cap \mc B(\ka)$. For any $b \in \mc A(\ka) \cap \mc B(\ka)$, the Shapovalov form
defines a nondegenerate pairing
\[S\colon\ \mc S^*(b,-\ka)\ox V/\mc S(b,\ka) \to \K
\]
of the fibers of these bundles, by Theorem \ref{thm orth}. The discrete connections on
$\mc S^*(\ka) \to \mc B(\ka)$ and~${\mc Q(\ka)\to \mc A(\ka)}$ are dual with respect to the Shapovalov form, that is, for any
$u\in \mc S^*(b,-\ka)$, $v\in V/\mc S(b,\ka)$, and $a=1,\dots,n$,
we have
$
S(u,v) = S( K_a(-b;-\ka)u, K_a(b ;\ka)v)$.

Define $p$-quasi-hypergeometric sections $T^\ell(z;\ka)$,
$\ell =1,\dots, d(-\ka)$, of the bundle $\mc Q(\ka)\to \mc A(\ka)$ over
$\mc A(\ka) \cap \mc B(\ka)$ by the formulas
\[
S\big(Q^{m p-1}(-b,-\ka), T^\ell(b,\ka)\big) = \dl_{\ell, m},
\qquad m=1,\dots, d(-\ka).
\]

 \begin{lem} \label{lem 6.1}
Assume that $p>n$ and
$0 < d(-\ka)\leq n-1$.
Then $p$-quasi-hypergeometric sections~${T^\ell(z;\ka)}$,
$\ell =1,\dots, d(-\ka)$, of the quotient bundle $\mc Q(\ka)\to \mc A(\ka) \cap \mc B(\ka)$
 form a flat basis of the space of sections of that bundle.
 \end{lem}

\begin{proof}
The proof is straightforward.
\end{proof}

If $d(-\ka)=0$, we define $\mc B(\ka) = \bar{\mc A}$. We also define $\mc S^*(\ka) \to \mc B(\ka)$
 to be the rank 0 subbundle of $V \times \mc B(\ka)$.

\begin{exmp}
Let $n=2$. Then $V$ is of dimension 1.
For $p=5$, $k=\ka=3$, we have $d(3) =1$, and the qKZ connection has a flat basis given by the $p$-hypergeometric solution
\[
Q^{4}(z_1,z_2)= (-2z_1 + 2z_2 +2, 2z_1 - 2z_2 - 2).
\]
For $p=5$. $k=\ka=2$, we have $d(2)=0$, and the qKZ connection has a flat basis given by the $p$-quasi-hypergeometric solution
\[
T^1(z_1,z_2, 2) = \left(\frac 3{2z_1 - 2z_2 +2}, \frac3{-2z_1 + 2z_2 - 2}\right) .
\]
\end{exmp}

\subsection[Reduced p-curvature operators]{Reduced $\boldsymbol{ p}$-curvature operators}

Let $\hat C_a(z;\ka)$, $a=1,\dots,n$, be the reduced $p$-curvature operators of the qKZ discrete connection
on $V\times \K^n \to \K^n$, see \eqref{red p-cu}.

\begin{thm}\label{thm p-cu}
If $p>n$, \smash{$\ka\in\F_p^\times$}, and $0< d(\ka) < n-1$,
then the span of the $p$-hypergeo\-metric sections
$Q^{\ell p-1}(z;\ka)$, $\ell = 1, \dots, d(\ka)$, lies in the kernel of~$\hat C_a(z,\ka)$ and contains the image of~$\hat C_a(z;\ka)$
for every $ a=1,\dots,n$.

If $d(\ka) =p-1$ or $0$, then all reduced $p$-curvature operators equal zero.
\end{thm}

\begin{cor}\label{cor p-cu}
We have
$
\hat C_a(z;\ka) \hat C_b(z;\ka) = 0$, $ a,b = 1,\dots, n$.
\end{cor}

\begin{proof}[Proof of Theorem \ref{thm p-cu}.]
The span lies in the kernel of $\hat C_a(z;\ka)$ by formula \eqref{fCf}.

The operator $\hat C_a(z;\ka)$ annihilates the span, hence $\hat C_a(z;\ka)$
induces a well-defined operator on the fibers of the quotient bundle
$\mc Q(\ka)\to \mc A(\ka)$. This induced operator is the $a$-th
reduced $p$-curvature operator of the qKZ connection on the quotient bundle.
The quotient bundle has a~flat basis of $p$-quasi-hypergeometric sections over the Zariski open subset
$\mc A(\ka) \cap \mc B(\ka)$. Hence all
reduced $p$-curvature operators of the qKZ connection on the quotient bundle are zero. Therefore the image of
$\hat C_a(z;\ka)$ is contained in the span.

If $d(\ka)=p-1$, then $\mc S(\ka) \to \mc A(\ka)$ coincides with
 $V\times \mc A(\ka)\to \mc A(\ka)$, and all reduced $p$-curvature operators
are zero by formula \eqref{fCf}.

If $d(\ka)=0$, then $p$-quasi-hypergeometric sections form a flat basis of the space of sections
of~${V\times \mc A(\ka)\to \mc A(\ka)}$, and again all reduced $p$-curvature operators
are zero by formula \eqref{fCf}.
\end{proof}

\begin{lem}\label{lem C ne 0}
If $p>n$, \smash{$\ka\in\F_p^\times$}, and $0< d(\ka) < n-1$,
 then every reduced $p$-curvature operator~${\tilde C_a(z;\ka)}$ is nonzero.
\end{lem}

\begin{proof}Consider a normalized $p$-curvature operator
$\tilde C_a^{\on{KZ}}(z,\ka)$
of the associated differential KZ equations. In a basis of $V$, the entries of the matrix of the
operator $\tilde C_a^{\on{KZ}}(z,\ka)$ are homogeneous polynomials in $z$ of degree $(n-2)p$.
By Corollary \ref{cor deg C}, in the same basis, the entries of the matrix of the
operator $\tilde C_a(z,\ka)$ are polynomials in $z$ of degree $(n-2)p$ whose top-degree parts equal
the corresponding entries of the matrix of the operator $\tilde C_a^{\on{KZ}}(z,\ka)$.

It is proved in \cite[Theorem 1.13]{VV1} that
if $p>n$, \smash{$\ka\in\F_p^\times$}, then every reduced $p$-curvature operator
$\tilde C_a^{\on{KZ}}(z,\ka)$ is a (nonzero) operator of rank 1. Hence every reduced $p$-curvature operator~$\tilde C_a(z,\ka)$ is a nonzero operator.
\end{proof}

\begin{exmp}
For $n=3$, we have $\dim V=2$. Let $p>3$ and $d(\ka)=1$. Then $d(-\ka)=1$.
For $a=1,2,3$, the kernel of the reduced $p$-curvature operator
$\hat C_a(z,\ka)$ is generated by $Q^{p-1}(z,\ka)$
and the image
of $\hat C_a(z,\ka)$ is generated by $Q^{p-1}(z,\ka)$. Such an operator is determined
uniquely up to multiplication
by a scalar rational function in~$z$.
\end{exmp}

For an operator $F\colon V\to V$, denote by $F^* \colon V\to V$ the operator dual to $F$ under the Shapovalov form,
$S( F^*x, y) =S(x, F y)$.

\begin{lem}\label{lem c=c}We have
\begin{gather}
\label{c=c}
\hat C_a(z,-\ka)
 = -
\hat C_a(-z;\ka)^*.
\end{gather}
\end{lem}

\begin{proof}
We have
$
S(x, y) = S( C_a(-z;-\ka) x, C_a(z;\ka) y)
$
by formulas \eqref{per} and \eqref{pm}. Hence
$
C_a(-z;-\ka) = \big(C_a(z;\ka)^{-1}\big)^*$.
We also have $(C_a(z;\ka)-1)^2=0$ by Corollary \ref{cor p-cu}. Then
\[
C_a(z;\ka)^{-1} = (1+ (C_a(z;\ka)-1))^{-1} = 1 - (C_a(z;\ka)-1)
\]
and $C_a(-z;-\ka) -1 = 1-C_a(z;\ka)^*$.
\end{proof}

\begin{cor}\label{cor c=c}
The normalized $p$-curvature operators satisfy the equation
\begin{gather}
\tilde C_a(z;-\ka)
= (-1)^n\tilde C_a(-z;\ka)^* .\label{tc=c}
\end{gather}
\end{cor}

\begin{proof}Formula \eqref{tc=c} follows from equation \eqref{c=c} and
the following formulas:
\begin{gather*}
\tilde C_a(z;\ka)=
(C_a(z;\ka)-1) \prod_{j\ne i} (z _i-z_j, \ka)_p,
\\
\tilde C_a(-z;-\ka)
=
(C_a(-z;-\ka)-1) \prod_{j\ne i} (-z _i+z_j, -\ka)_p .\tag*{\qed}
\end{gather*}\renewcommand{\qed}{}
\end{proof}

\subsection[All solutions of qKZ equations for ka in F\_p\^times]{All solutions of qKZ equations for $\boldsymbol{\ka \in \F_p^\times}$}
\label{sec 7.4}

\begin{thm}\label{thm main}
Let $p>n$, $\ka\in\F_p^\times$, and $0<d(\ka)<n-1$.
Let $f(z)$ be a $V$-valued rational function in $z$ which is a solution of
 the qKZ equations with step $\ka$.
Then $f(z)$ is a linear combination of the $p$-hypergeometric solutions
$Q^{\ell p-1}(z;\ka)$, $\ell = 1, \dots, d(\ka)$, with coefficients which are rational functions
in $z_i^p-z_i$, $i=1,\dots, n$.
\end{thm}

Recall that if $d(\ka) = n-1$, then the qKZ connection has a basis of flat sections given by
 the $p$-hypergeometric sections by Theorem \ref{thm ind}, and
if $d(\ka) = 0$, then the qKZ connection has a~basis of flat sections given
 by the $p$-quasi-hypergeometric solutions, by Lemma \ref{lem 6.1}.

\begin{proof}For $a=1,\dots,n$, consider the normalized $p$-curvature operators
$\tilde C_a^{\on{KZ}}(z,\ka)$ and $\tilde C_a(z,\ka)$. Both of these operators are polynomials in $z$,
and the polynomial $\tilde C_a^{\on{KZ}}(z,\ka)$ is the top-degree part of the polynomial
$\tilde C_a(z,\ka)$. The polynomial $\tilde C_a^{\on{KZ}}(z,\ka)$ is nonzero by \cite[Theorem 1.13]{VV1}
and hence $\tilde C_a(z,\ka)$ is a nonzero operator.

It was proved in \cite[Theorem 1.8]{VV1} that if $p>n$, $\ka\in \F_p^\times$, and
$0<d(\ka)<n-1$, then all solutions of the KZ equations are linear combinations of the
$p$-hypergeometric solutions. Hence the intersection of kernels of the operators $\tilde C_a^{\on{KZ}}(z,\ka)$, $a=1,\dots, n$,
is of dimension~$d(\ka)$ for generic~$z$, and the span of images of the operators
$\tilde C_a^{\on{KZ}}(z,\ka)$, $a=1,\dots, n$, is of dimension~${n-1-d(\ka)}$ for generic $z$. Therefore, the span of images of the operators
$\tilde C_a(z,\ka)$, $a=1,\dots, n$, has dimension at least $n-1-d(\ka)$ for generic $z$.
This implies that the span of values of flat sections of the qKZ connection is of dimension
not larger than $d(\ka)$ for generic $z$. But we have $d(\ka)$ flat linear independent $p$-hypergeometric sections
$Q^{\ell p-1}(z;\ka)$, $\ell = 1, \dots, d(\ka)$. Hence any flat section of the qKZ connection is a linear combination of the $p$-hypergeometric sections with 1-periodic coefficients.
\end{proof}

\begin{cor}\label{cor ort}
Let $p>n$, $\ka\in\F_p^\times $, and $0<d(\ka)<n-1$.
Let $f(z)$ and $g(z)$ be $V$-valued rational functions in $z$
where $f(z)$ is a solution
 the qKZ equations with step $\ka$ and $g(z)$ is a~solution
 the qKZ equations with step $-\ka$. Then
\begin{gather}
 \label{ort fg}
 S(g(-z), f(z))=0.
\end{gather}
\end{cor}

Formula \eqref{ort fg} follows from Theorems \ref{thm main}
and \ref{thm orth}.

\subsection[qKZ connection with ka in K setminus F\_p]{qKZ connection with $\boldsymbol{ \ka \in \K\setminus \F_p}$}

\begin{lem}\label{lem ne F}
Let $p>n$ and $\ka\in\K\setminus \F_p$ . Then all the normalized $p$-curvature operators~${\tilde C_a(z;\ka)}$,
$a=1,\dots,n$, are nondegenerate for generic $z$.
\end{lem}

\begin{proof}
Formula (3.19) in \cite{VV1} describes the spectrum of the $p$-curvature operators $C_a^{\on{KZ}}(z;\ka)$
of differential KZ equations. The formula shows that all $p$-curvature operators $C_a^{\on{KZ}}(z;\ka)$
are nondegenerate for generic $z$.
In a basis of $V$, the matrices of $\tilde C_a^{\on{KZ}}(z;\ka)$ are homogeneous polynomials
in $z$ of degree $(n-2)p$. Hence their determinants are nonzero homogeneous polynomials in $z$. By Corollary \ref{cor deg C},
the determinants of the normalized $p$-curvature operators $\tilde C_a(z;\ka)$ are nonzero polynomials. The lemma follows.
\end{proof}

\begin{cor}\label{cor ne F}
For $p>n$ and $\ka\in\K\setminus \F_p$, there does not exist a nonzero rational $V$-valued function $f(z)$ which is a flat section
of the qKZ connection with parameter $\ka$.
\end{cor}

\begin{proof}
If $f(z)$ is a flat section, then it lies in the kernel of every normalized $p$-curvature operator~$\tilde C_a(z;\ka)$. That contradicts to Lemma \ref{lem ne F}.
\end{proof}

\subsection*{Acknowledgements}
E.M. was supported in part by Simons Foundation grant 709444. A.V.\ was supported in part by Simons Foundation grant TSM-00012774. The authors thank P.~Etingof, R.~Rimanyi, and V.~Vologodsky for useful discussions.

\pdfbookmark[1]{References}{ref}
\LastPageEnding

\end{document}